\documentclass[11pt]{article}
\usepackage{amsmath}
\usepackage{amsfonts,times}
\usepackage{cleveref}
 \usepackage{times} 
 
\usepackage{url}
\usepackage{color,microtype}
\usepackage{graphicx,wrapfig,subfigure,titling}
\usepackage{amssymb,amsfonts}
\usepackage{amsthm}
\usepackage{textcomp}
\usepackage{multirow}
\usepackage{gensymb}
\usepackage{array}
\usepackage{times}
\usepackage[sort]{cite}

\usepackage{fullpage}
\usepackage{amsmath,amssymb}

\newtheorem{theorem}{Theorem}

\newtheorem{corollary}[theorem]{Corollary}
\newtheorem{lemma}[theorem]{Lemma}

\newcommand{\expec}[1]{\mathbb E\left [ #1 \right ]}

\DeclareMathOperator{\polylog}{polylog}
\DeclareMathOperator{\poly}{poly}

\DeclareMathOperator{\argmax}{{argmax}}

\newcommand{\prob}[1]{\mathbb P \left [ #1 \right ]}

\newcommand{\bin}{{\mathbf{Bin}}}
\newcommand{\eat}[1]{}

\newcommand{\A}{{\cal A}}

\newcommand{\bx}{\mathbf{x}}
\newcommand{\bi}{\mathbf{i}}

\date{}

\begin{document}
\title{Densest Subgraph in Dynamic Graph Streams\thanks{University of Massachusetts, Amherst. 
\texttt{\{mcgregor,dtench,svorotni,hvu\}@cs.umass.edu}. This work was supported by NSF  Awards CCF-0953754, IIS-1251110,  CCF-1320719, and a Google Research Award.}}
\author{
Andrew McGregor
 \and 
David Tench
 \and 
Sofya Vorotnikova
 \and 
Hoa T.~Vu }
\maketitle

\begin{abstract}
In this paper, we consider the problem of approximating the densest subgraph in the dynamic graph stream model. In this model of computation, the input graph is defined by an arbitrary sequence of edge insertions and deletions and the goal is to analyze properties of the resulting graph given memory that is sub-linear in the size of the stream.
We present a single-pass algorithm that returns a $(1+\epsilon)$ approximation of the maximum density with high probability; the algorithm uses $O(\epsilon^{-2} n \polylog n)$ space, processes each stream update in $\polylog (n)$ time, and uses $\poly(n)$ post-processing time where $n$ is the number of nodes.  The space used by our algorithm matches the lower bound of Bahmani  et al.~(PVLDB 2012) up to a poly-logarithmic factor for constant $\epsilon$. The best existing results  for this problem were established recently by Bhattacharya et al.~(STOC 2015). They presented a $(2+\epsilon)$ approximation algorithm using similar space and another algorithm that both processed each update and maintained a $(4+\epsilon)$ approximation of the current maximum density in $\polylog (n)$ time per-update. 
\end{abstract}

\section{Introduction}
In the dynamic graph stream model of computation, a sequence of edge insertions and deletions defines an input graph and the goal is to solve a specific problem on the resulting graph given only one-way access to the input sequence and limited working memory. Motivated by the need to design efficient algorithms for processing massive graphs, over the last four years there has been a considerable amount of work designing algorithms in this model\cite{AhnCGMW15,AhnGM12a,AhnGM12b,AhnGM13,KapralovLMMS14,KapralovW14,GoelKP12,KutzkovP14a,GuhaMT15,BhattacharyaHNT15,ChitnisCEHMMV15,AssadiKLY15,Konrad15,BuryS15a}. Specific results include testing edge connectivity \cite{AhnGM12b} and node connectivity \cite{GuhaMT15}, constructing spectral sparsifiers \cite{KapralovLMMS14}, approximating the densest subgraph \cite{BhattacharyaHNT15}, maximum matching \cite{ChitnisCEHMMV15,AssadiKLY15,Konrad15,BuryS15a}, correlation clustering \cite{AhnCGMW15}, and estimating the number of triangles \cite{KutzkovP14a}. For a recent survey of the area, see \cite{McGregor14}. 

In this paper, we consider the densest subgraph problem.
Let $G_U$ be the induced subgraph of graph $G=(V,E)$ on nodes $U$. Then the \emph{density} of $G_U$ is defined as \[d(G_U)=|E(G_U)|/|U| \ ,\] where $E(G_U)$ is the set of edges in the induced subgraph. We define the \emph{maximum density} as 
\[
d^* = \max_{U\subseteq V} d(G_U) \ .
\]
and say that the corresponding subgraph is the \emph{densest subgraph}. The densest subgraph can be found in polynomial time 
\cite{GalloGT89,Goldberg84,Charikar00,KhullerS09} and more efficient approximation algorithms have been designed \cite{Charikar00}. Finding dense subgraphs is an important primitive when analyzing massive graphs; applications include community detection in social networks and identifying link spam on the web, in addition to applications on financial and biological data.  See \cite{densesurvey} for a survey of applications and existing algorithms for the problem.

\subsection{Our Results and Previous Work}

We present a single-pass algorithm that returns a $(1+\epsilon)$ approximation with high probability\footnote{Throughout this paper, we say an event holds with high probability if the probability is at least $1-n^{-c}$ for some constant $c>0$.}. For a graph on $n$ nodes, the algorithm uses the following resources:
\begin{itemize} 
\item {\em Space:} $O(\epsilon^{-2} n \polylog n)$. The space used by our algorithm matches the lower bound of Bahmani  et al.~\cite{BahmaniKV12} up to a poly-logarithmic factor for constant $\epsilon$. 
\item {\em Per-update time:} $\polylog (n)$. We note that this is the worst-case update time rather than amortized over all the edge insertions and deletions. 
\item {\em Post-processing time:} $\poly(n)$. This will follow by using any exact algorithm for densest subgraph \cite{GalloGT89,Goldberg84,Charikar00} on the subgraph  generated  by our algorithm.
\end{itemize}  

The most relevant previous results for the problem were established recently by Bhattacharya et al.~\cite{BhattacharyaHNT15}. They presented two algorithms that use similar space to our algorithm and process updates in $\polylog (n)$ amortized time. The first algorithm returns a $(2+\epsilon)$ approximation of the maximum density of the final graph while the second (the more technically challenging result) outputs a $(4+\epsilon)$ approximation of the current maximum density after every update while still using only $\polylog (n)$ time per-update. 
Our algorithm improves the approximation factor to $(1+\epsilon)$ while keeping the same space and update time. It is possible to modify our algorithm to output a $(1+\epsilon)$ approximation to the current maximum density after each update but the simplest approach would require the post-processing step  to be run after every edge update and this  would not be efficient.

Bhattacharya et al.~were one of the first to combine the space restriction of graph streaming with the fast update and query time requirements of fully-dynamic algorithms from the dynamic graph algorithms community. 
Epasto, Lattanzi, and Sozio \cite{EpastoLS15} present a fully-dynamic algorithm that returns a $(2+\epsilon)$ approximation of the current maximum density.
Other relevant work includes papers by Bahmani, Kumar, and Vassilvitskii~\cite{BahmaniKV12} and Bahmani, Goel, and  Munagala \cite{BahmaniGM14}. The focus of these papers is on designing algorithms in the MapReduce model but the resulting algorithms can also be implemented in the data stream model if we allow multiple passes over the data.

\subsection{Our Approach and Paper Outline}

The approach we take in this paper is as follows. In Section \ref{sec:subsample}, we show that if we sample every edge of a graph independently with a specific probability then we generate a graph that is a) sparse and b) can be used to estimate the maximum density of the original graph. This is not difficult to show but requires care since there are an exponential number of subgraphs in the subsampled graph that we will need to consider.

In Section \ref{sec:streamalg}, we show how to perform this sampling in the dynamic graph stream model. This can be done using the $\ell_0$ sampling primitive \cite{JowhariST11,CormodeF14} that enables edges to be sampled uniformly from the set of edges that have been inserted but not deleted. However, a naive application of this primitive would necessitate $\Omega(n)$ per-update processing. To reduce this to $O(\polylog n)$ we reformulate the sampling procedure in such a way that it can be performed more efficiently. This reformulation is based on creating multiple partitions of the set of edges using pairwise independent hash functions and then sampling edges within each group in the partition. The use of multiple partitions is somewhat reminiscent of that used in the Count-Min sketch \cite{CormodeM05}.

\paragraph{Remark.} Independently of our work,  Esfandiari, Hajiaghayi, and Woodruff \cite{EsfandiariHW15}  also proved a similar result to that presented in this paper. Their result is also based on uniformly sampling edges but their approach for ensuring fast update time differs and may be of independent interest.

\section{Subsampling Approximately Preserves Maximum Density}\label{sec:subsample}
In the section, we consider properties of a random subgraph of the input graph $G$.
Specifically, let $G'$ be the graph formed by sampling each edge in $G$ independently with probability $p$ where
\[p = c \epsilon^{-2} \log n \cdot \frac{n}{m}\]
for some sufficiently large constant $c>0$ and $0<\epsilon<1/2$. We may assume that $m$ is sufficiently large such that $p<1$ because otherwise we can reconstruct the entire graph in the allotted space using standard results from the sparse recovery literature \cite{GilbertI10}.

We will prove that, with high probability, the maximum density  of $G$ can be estimated up to factor $(1+\epsilon)$ given $G'$. While it is easy to analyze how the density of a specific subgraph changes after the edge sampling, we will need to consider all $2^n$ possible induced subgraphs and prove properties of the subsampling for all of them. 

The next lemma shows that $d(G'_U)$ is roughly proportional to $d(G_U)$ if $d(G_U)$ is ``large" whereas if $d(G_U)$ is ``small" then $d(G'_U)$ will also be relatively small.

\begin{lemma}
\label{lem:ld}
Let $U$ be an arbitrary set of $k$ nodes. Then, 
\begin{align*}
\prob{d(G'_U)\geq pd^*/10} & \leq  n^{-10k}    & & \mbox{if $d(G_U)\leq d^*/60$}\\
\prob{|d(G'_U)-pd(G_U)|  \geq  \epsilon p d(G_U) } & \leq  2n^{-10k}     & & \mbox{if $d(G_U)> d^*/60$} \ . 
\end{align*}

\end{lemma}
\begin{proof}
We start by considering the density of the entire graph $d(G)=m/n$ and therefore conclude that the maximum density, $d^*$, is at least $m/n$. Hence, $p\geq  (c \epsilon^{-2} \log n)/d^*$.

Let $X$ be the number of edges in $G'_U$ and note that $\expec{X}= pkd(G_U)$. First assume $d(G_U)\leq d^*/60$. Then, by an application of the Chernoff Bound (e.g., \cite[Theorem 4.4]{MitzenmacherE05}), we observe that
\begin{eqnarray*}
\prob{d(G'_U)\geq  p d^*/10} 
= \prob{X\geq  pkd^*/10} 
\leq  2^{-pkd^*/10} 
<  2^{-c k(\log n)/10} 
\end{eqnarray*}
and this is at most $n^{-10k}$ for sufficiently large constant $c$. 

Next assume $d(G_U)> d^*/60$. Hence, by an application of an alternative form of the Chernoff Bound (e.g., \cite[Theorem 4.4 and 4.5]{MitzenmacherE05}), we observe that
\begin{eqnarray*}
\prob{|d(G'_U)-pd(G_U)| \geq \epsilon pd(G_U) } &=& 
\prob{|X-pkd(G_U)|  \geq  \epsilon pkd(G_U)} \\
& \leq & 2\exp(- \epsilon^2 pkd(G_U)/3)  \\
& \leq & 2\exp(- \epsilon^2 pkd^*/180) \\
& \leq & 2\exp(- ck (\log n) /180) \ .
\end{eqnarray*}
and this is at most $2n^{-10k}$ for sufficiently large constant $c$.
\end{proof}

\begin{corollary}
\label{cor:ld}
With high probability, for all $U\subseteq V$:
\[d(G'_U)\geq (1-\epsilon) pd^* ~~ \Rightarrow~~ d(G_U)\geq \frac{1-\epsilon}{1+\epsilon} \cdot d^* \ .\]
\end{corollary}
\begin{proof}
There are ${n \choose k}\leq n^k$ subsets of $V$ that have size $k$. Hence, by appealing to Lemma~\ref{lem:ld} and the union bound, with probability at least $1-2n^{-9k}$, the following two equations hold,
\begin{align*}
d(G'_U) \geq  pd^*/10 
 & ~~\Rightarrow~~ d(G_U)>  d^*/60 
\\
d(G_U)>  d^*/60 
 & ~~\Rightarrow~~  d(G_U)\geq \frac{d(G'_U)}{p(1+\epsilon)} 
\end{align*}
for all $U\subseteq V$ such that $|U|=k$. Since  $(1-\epsilon) pd^*\geq pd^*/10$, together these two equations imply 
\[
d(G'_U)\geq (1-\epsilon) pd^* ~~ \Rightarrow~~ d(G_U)\geq \frac{d(G'_U)}{p(1+\epsilon)} \geq \frac{1-\epsilon}{1+\epsilon} \cdot d^*
\]
for all sets $U$ of size $k$. Taking the union bound over all values of $k$ establishes the corollary.\end{proof}

We next show that the densest subgraph in $G'$ corresponds to a subgraph in $G$ that is almost as dense as the densest subgraph in $G$.

\begin{theorem}
Let $U'=\argmax_U d(G'_U)$. Then with high probability, 
\[
 \frac{1 - \epsilon}{1+\epsilon} \cdot d^* 
\leq d(G_{U'}) \leq d^* \ .
\]
\end{theorem}
\begin{proof}
Let $U^*= \argmax_U d(G_U)$. By appealing to Lemma~\ref{lem:ld}, we know that $d(G'_{U^*})\geq (1 - \epsilon) p d^*$ with high probability. Therefore 
\begin{equation*}
d(G'_{U'}) \geq d(G'_{U^*}) \geq (1 - \epsilon) p d^* \ ,
\end{equation*}
and the result follows by appealing to Corollary \ref{cor:ld}.
\end{proof}

\section{Implementing in the Dynamic Data Stream Model}\label{sec:streamalg}

In this section, we show how to sample each edge independently with the prescribed probability in the dynamic data stream model. The resulting algorithm  uses $O(\epsilon^{-2} n\polylog n)$ space. The near-linear dependence on $n$ almost matches the $\Omega(n)$ lower bound proved by Bahmani  et al.~\cite{BahmaniKV12}. The main theorem we prove is:

\begin{theorem}\label{thm:bigthm}
There exists a randomized algorithm in the dynamic graph stream model that returns a $(1+\epsilon)$-approximation for the density of the densest subgraph with high probability. The algorithm uses $O(\epsilon^{-2} n\polylog n)$ space and $O(\polylog n)$ update time. The post-processing time of the algorithm is polynomial in $n$.
\end{theorem}

To sample the edges with probability $p$ in the dynamic data stream model there are two main challenges: 
\begin{enumerate}
\item Any edge we sample during the stream may subsequently be deleted. 
\item Since $p$ depends on $m$, we do not know the value of $p$ until the end of the stream. 
\end{enumerate}
To address the first challenge, we appeal to an existing result on the $\ell_0$ sampling technique \cite{JowhariST11}: there exists an algorithm using $\polylog(n)$ space and update time that returns an edge chosen uniformly at random from the final set of edges in the graph. Consequently we may sample $r$ edges uniformly at random using $O(r\polylog n)$ update time and space. To address the fact we do not know $p$ apriori, we could set $r\gg pm=c\epsilon^{-2} n \log n$, and then, at the end of the stream when $p$ and $m$ are known a) choose $X \sim \bin(m,p)$ where  $\bin(\cdot,\cdot)$ denotes the binomial distribution and b) randomly pick $X$ distinct random edges amongst  the set of $r$ edges sampled (ignoring duplicates). This approach will work with high probability if $r$ is sufficiently large since $X$ is tightly concentrated around $\expec{X}=pm$.
However, a naive implementation of this algorithm would require $\omega(n)$ update time. The main contribution of this section is to demonstrate how to ensure $O(\polylog n)$ update time.

\subsection{Reformulating the Sampling Procedure}

We first describe an alternative sampling process that, with high probability, returns a set of edges $S$ where each edge in $S$ has been sampled independently with probability $p$ as required. The purpose of this alternative formulation is that it will allow us to argue that it can be emulated in the dynamic graph stream  model efficiently.

\paragraph{Basic Approach.} 
The basic idea is to partition the set of edges into different groups and then sample edges within groups that do not contain too many edges. We refer to such groups as ``small". We determine which of the edges in a small group are to be sampled in two steps:
\begin{itemize}
\item {\em Fix the number $X$ of edges to sample:} Let $X\sim \bin(g,p)$ where $g$ is the number of edges in the relevant group.
\item {\em Fix which $X$ edges to sample:} We then randomly pick $X$ edges without replacement from the relevant group.
\end{itemize}
It is not hard to show that this two-step process ensures that each edge in the group is sampled independently with probability $p$. At this point, the fate of all edges in small groups has been decided: they will either be returned in the final sample or definitely not returned in the final sample. 

We next consider another partition of the edges and again consider groups that do not contain many edges. We then determine the fate of the edges in such groups whose fate has not hitherto  been  determined. We keep on considering different partitions until every edge has been included in a small group and has had its fate determined. 

\begin{lemma}
Assume for every edge there exists a partition such that the edge is in a small group. Then the distribution over sets of sampled edges is the same as the distribution had each edge been sampled independently with probability $p$.
\end{lemma}
\begin{proof}
The proof does not depend on the exact definition of ``small" and the only property of the partitions that we require is that every edge is in a small group of some partition. We henceforth consider a fixed set of partitions with this property.

We first consider the $j$th group in the $i$th partition. Let $g$ be the number of edges in this group. For any subset $Q$ of $\ell$ edges in this group, we show that the probability that $Q$ is picked by the two-step process above is indeed $p^\ell$. 
%
%
\begin{align*}
 \prob{ \forall e \in Q, \mbox{$e$ is picked} }   &= \sum_{t=\ell}^g \prob{\forall e \in Q, \mbox{$e$ is picked } | X = t} \prob{X=t} \\
& = \sum_{t = \ell}^g \frac{ { g-\ell \choose t- \ell} } {{g \choose t}} \cdot {g \choose t} \cdot p^t (1-p)^{g-t} \\
& = p^\ell \sum_{t = \ell}^g { g-\ell \choose t- \ell}  \cdot p^{t-\ell} (1-p)^{g-t} = p^\ell.
\end{align*}
and hence  edges within the same group are sampled independently with probability $p$. Furthermore, the edges in different groups of the same partition are sampled independently from each other.
 
Let $f(e)$ be the first partition in which $e$ is placed in a group that is small and let $W_i=\{e:f(e)=i\}$. 
Restricting $Q$ to edges in $W_i$ in the above analysis establishes that edges in each $W_i$ are sampled independently. Since $f(e)$ is determined by the fixed set  of partitions rather than the randomness of the sampling procedure, we also conclude that edges in different $W_i$ are sampled independently.  As we assume that every edge belongs to at least one small group in some partition, if we let $r$ be the total number of partitions, then $\{W_i \}_{i \in [r] }$ partition the set of edges $E$. Hence, all edges in $E$ are sampled independently with probability $p$.
\end{proof}

\paragraph{Details of Alternative Sampling Procedure.}
The partitions considered will be determined by pairwise independent hash functions and we will later argue that it is sufficient to consider only $O(\log n)$ partitions. Each hash function will partition the $m$ edges into $n\epsilon^{-2}$ groups. In expectation the number of edges in a group will be $\epsilon^{2} m/n$ and we define a  group to be small if it contains at most $t=4\epsilon^{2} m/n$ edges. We therefore expect to sample less than  $4p\epsilon^{2} m/n=4c\log n$ edges from a small group. We will abort the algorithm if we attempt to sample significantly more edges than this from some small group. The procedure is as follows:

\begin{itemize}
\item Let $h_1,\ldots ,h_r:{n \choose 2} \rightarrow [n\epsilon^{-2}]$ be pairwise independent hash functions where $r=10 \log n$.
\item  Each $h_i$ defines a partition of $E$ comprising of sets of the form
 \[E_{i,j}=\{e\in E: h_i(e)=j\} \ .\] 
Say $E_{i,j}$ is \emph{small} if it is of size at most $t=4\epsilon^{2}m/n$. Let $D_i$ be the set of all edges in the small sets determined by $h_i$.
\item For each small $E_{i,j}$, let 
\[X_{i,j}=\bin(|E_{i,j}|,p)\] and abort if 
\[X_{i,j} \geq \tau ~~~ \mbox{ where } ~~~ \tau =24 c\log n \ .\]
Let $S_{i,j}$ be a set of $X_{i,j}$ edges sampled without replacement from $E_{i,j}$.

\item Let $S$ be set of edges that were sampled among some $D_i$ that are not in $D_1\cup D_2 \cup \ldots \cup D_{i-1}$, i.e., edges whose fate had not already been determined.
\[
S=\bigcup_{i=1}^{r} \{e\in D_i: e\in \cup_j S_{i,j} \mbox{ and } e\not \in D_1\cup D_2 \cup \ldots \cup D_{i-1}\}
\]
\end{itemize}

\paragraph{Analysis.} There are two main things that we need to show to establish that the above process emulates our basic sampling approach with high probability. First, we will show that with high probability for every edge $e$ there exists $i$ and $j$ such that $e\in E_{i,j}$ and $E_{i,j}$ is small. This ensures that we will make a decision on whether $e$ is included in the final sample. Second, we will show that it is very unlikely we abort because some $X_{i,j}$ is too large.

\begin{lemma}
With probability at least $1-n^{-8}$, for every edge $e$ there exists $i$ such that $e\in E_{i,j}$ and $E_{i,j}$ is small. \end{lemma}
\begin{proof}
Fix $i\in [r]$ and let $j=h_i(e)$. Then $\expec{|E_{i,j}|}\leq 1+\epsilon^{2} (m-1)/n\leq 2\epsilon^{2} m/n$ assuming $m\geq \epsilon^{-2}n$. By an application of the Markov bound:
\[
\prob{|E_{i,j}|\geq 4m\epsilon^{2}/n}\leq 1/2 \ .
\] 
Since each $h_i$ is independent,
\[
\prob{|E_{i,h_i(e)}|\geq 4m\epsilon^{2}/n \mbox{ for all $i$} }\leq 1/2^r =  1/n^{10} \ .
\]  
Therefore by the union bound over all $m\leq n^2$ edges there exists a good partition for each $e$ with probability at least $1-n^{-8}$.
\end{proof}

\begin{lemma}
With high probability, all $X_{i,j}$ are less than $\tau=24 c  \log n$.
\end{lemma}
\begin{proof}
Since $E_{i,j}$ is small then $\expec{X_{i,j}}=|E_{i,j}|p \leq 4\epsilon^{2}p  m/n  = 4c \log n$. Hence, by an application of the Chernoff bound, 
\[\prob{X_{i,j}\geq 24 c \log n} \leq 2^{-24 c \log n} \leq n^{-10} \ .\] Taking the union bound over all $10\log n$ values of $i$ and $\epsilon^{-2}n$ values of $j$ establishes the lemma. 
\end{proof}

\subsection{The Dynamic Graph Stream Algorithm}

We are now ready to present the dynamic graph stream algorithm. To emulate the above sampling process in the dynamic graph stream model, we proceed as follows:
\begin{enumerate}
\item {\em Pre-Processing:} Pick the hash functions $h_1, h_2, \ldots, h_r$. These define the sets $E_{i,j}$.
\item {\em During One Pass:} 
\begin{itemize}
\item Compute the size of each $E_{i,j}$ and $m$. Note that $m$ is necessary to define $p$.
\item Sample $\tau $ edges $S'_{i,j}$ uniformly without replacement from each $E_{i,j}$. 
\end{itemize}
\item {\em Post-Processing:}
\begin{itemize}
\item Randomly determine the values $X_{i,j}$ based on the exact values of $|E_{i,j}|$ and $m$ for each $E_{i,j}$ that is small. If $X_{i,j}$ exceeds $\tau$ then abort.
\item Let $S_{i,j}$ be a random subset of $S_{i,j}'$ of size $X_{i,j}$.
\item Return $p^{-1}\max_U d(G'_U)$ where $G'$ is the graph with edges:
\[
S=\bigcup_{i=1}^{r} \{e\in D_i: e\in \cup_j S_{i,j} \mbox{ and } e\not \in D_1\cup D_2 \cup \ldots \cup D_{i-1}\}
\]
\end{itemize}
\end{enumerate}

Note that is possible to compute  $|E_{i,j}|$ using a  counter that is incremented or decremented whenever an edge $e$ is added or removed respectively that satisfies $h_i(e)=j$. We may evaluate pairwise independent hash functions in $O(\polylog n)$ time. The exact value of $\max_U d(G'_U)$ can be determined in polynomial time using the result of Charikar \cite{Charikar00}. To prove Theorem~\ref{thm:bigthm}, it remains to describe how to sample $\tau$ edges \emph{without replacement} from each $E_{i,j}$.

\paragraph{Sampling Edges Without Replacement Via $\ell_0$-Sampling.} To do this, we use the $\ell_0$-sampling algorithm of Jowhari et al.~\cite{JowhariST11}. Their algorithm returns, with high probability, a random edge from $E_{i,j}$ and the space and update time of the algorithm are both $O(\polylog n)$. Running $\tau$ independent instantiations of this algorithm immediately enables us to sample $\tau$ edges uniformly from $E_{i,j}$ \emph{with replacement}. 

However, since their algorithm is based on linear sketches, there is an elegant way (at least, more elegant than simply over sampling and removing duplicates) to ensure that all samples are distinct. Specifically, let $\bx$ be the characteristic vector of the set $E_{i,j}$. Then, $\tau$ instantiations of the algorithm of Jowhari et al.~\cite{JowhariST11} generate random projections  \[\A_1 (\bx)~, ~\A_2 (\bx)~,~ \ldots~,~ \A_\tau (\bx)\] of $\bx$ such that a random non-zero entry of $\bx$ (which corresponds to an edge from $E_{i,j}$) can be identified by processing each $\A_i (\bx)$. Let $e_1$ be the edge reconstructed from $\A_1( \bx)$. Rather than reconstructing an edge from $\A_2 (\bx )$, which could be the same as $e_1$, we instead reconstruct an edge $e_2$ from 
\[\A_2 (\bx)-\A_2 (\bi_{e_1})=\A_2 (\bx-\bi_{e_1})\] where $\bi_{e_1}$ is the characteristic vector of the set $\{e_1\}$. Note that $e_2$ is necessarily different from $e_1$ since $\bx-\bi_e$ is the characteristic vector of the set $E_{i,j}\setminus \{e_1\}$. Similarly we reconstruct $e_j$ from 
\[\A_j (\bx)-\A_j (\bi_{e_1})-\A_j (\bi_{e_2})-\ldots - \A_j (\bi_{e_{j-1}})=\A_2 (\bx-\bi_{e_1}-\ldots - \bi_{e_{j-1}})\]
and note that $e_j$ is necessarily distinct from $\{e_1, e_2,\ldots, e_{j-1}\}$.

\section{Conclusion}
We presented the first algorithm for estimating the density of the densest subgraph up to a $(1+\epsilon )$ factor  in the dynamic graph stream model. Our algorithm used $O(\epsilon^{-2} n \polylog n)$ space, $\polylog (n)$ per-update processing time, and $\poly(n)$ post-processing to return the estimate. The most relevant previous results, by Bhattacharya et al.~\cite{BhattacharyaHNT15}, were a $(2+\epsilon)$ approximation in similar space and a $(4+\epsilon)$ approximation with $\polylog (n)$ per-update processing time that also outputs an estimate of the maximum density after each edge insertion or deletion. A natural open question is whether it is possible to use ideas contained in this paper to improve the approximation factor for the problem of maintaining a running estimate of the maximum density.

{ \small
\bibliographystyle{abbrv} \bibliography{dynamic}
}
\end{document}